\documentclass[lettersize,journal]{IEEEtran}
\usepackage{graphicx}
\usepackage[colorlinks,linkcolor=black,anchorcolor=black,citecolor=black]{hyperref}
\usepackage{lineno,hyperref}
\modulolinenumbers[5]
\usepackage[cmex10]{amsmath}
\usepackage{amsfonts}

\usepackage{color}
\usepackage{algorithm}
\usepackage{algorithmic}
\usepackage{array}
\usepackage{multirow}
\usepackage{booktabs}
\usepackage{bm}

\allowdisplaybreaks[4]
\usepackage{amsmath}
\usepackage{underscore}
\usepackage{stfloats}
\usepackage{float}
\usepackage{lipsum}
\usepackage{cite}
\usepackage{subfigure}
\usepackage{adjustbox}
\usepackage{amsthm}
\newtheorem{remark}{Remark}
\newtheorem{lemma}{Lemma}
\newtheorem{theorem}{Theorem}
\newtheorem{definition}{Definition}

\makeatletter

\usepackage{amssymb}
\usepackage{amsmath}
\begin{document}

\title{Cross-Layer Security for Semantic Communications: Metrics and Optimization}

\author{Lingyi Wang,
Wei Wu,~\IEEEmembership{Member,~IEEE,}
Fuhui Zhou,~\IEEEmembership{Senior Member,~IEEE,}\\
Zhijin Qin,~\IEEEmembership{Senior Member,~IEEE,}
Qihui Wu,~\IEEEmembership{Fellow,~IEEE}

\thanks{Copyright (c) 2025 IEEE. Personal use of this material is permitted. However, permission to use this material for any other purposes must be obtained from the IEEE by sending a request to pubs-permissions@ieee.org.}
\thanks{Lingyi Wang is with the College of Science,
Nanjing University of Posts and Telecommunications, Nanjing, 210003, China.
(e-mail: lingyiwang@njupt.edu.cn).}
\thanks{
Wei Wu is with the College of Communication and Information Engineering,
Nanjing University of Posts and Telecommunications, Nanjing, 210003, China,
and also with the National Mobile Communications Research Laboratory, Southeast University, Nanjing, 210096, China, and the Anhui Province Key Laboratory of Cyberspace Security Situation Awareness and Evaluation, China.
(e-mail: weiwu@njupt.edu.cn).}
\thanks{Fuhui Zhou and Qihui Wu are with the College of Electronic and Information Engineering, 
Nanjing University of Aeronautics and Astronautics, Nanjing, 210000, China.
(e-mail: zhoufuhui@ieee.org, wuqihui2014@sina.com).}
\thanks{Zhijin Qin is with the Department of Electronic Engineering, Tsinghua University, Beijing, 100084, China,
and also with the Beijing National Research Center for Information Science and Technology, Beijing, China, and the State Key Laboratory of Space Network and Communications, Beijing, China.
(e-mail: qinzhijin@tsinghua.edu.cn).}
}

\maketitle
\begin{abstract}
  Different from traditional secure communication that focuses on symbolic protection at the physical layer, 
  semantic secure communication requires further attention to semantic-level task performance at the application layer. 
  There is a research gap on how to comprehensively evaluate and optimize the security performance of semantic communication.
  In order to fill this gap, a unified semantic security metric, the cross-layer semantic secure rate (CLSSR), is defined to estimate cross-layer security requirements at both the physical layer and the application layer. 
  Then, we formulate the maximization problem of the CLSSR with the mixed integer nonlinear programming (MINLP). 
  We propose a hierarchical AI-native semantic secure communication network with a reinforcement learning (RL)-based semantic resource allocation scheme,
  aiming to ensure the cross-layer semantic security (CL-SS). 
  Finally, we prove the convergence of our proposed intelligent resource allocation, 
  and the simulation results demonstrate that our proposed CLSS method outperforms the traditional physical layer semantic security (PL-SS) method in terms of both task reliability and CLSSR.
\end{abstract}

% keywords

\IEEEpeerreviewmaketitle

\section{Introduction}
The artificial intelligence (AI)-native wireless network is considered as a promising communication technology 
for the next and even future generations since it is highly reliable, low-latency and actively cognitive \cite{saad2024artificial}.
Different from the Shannon communication framework that focuses only on the effectiveness of bits at the physical layer, 
the AI-native wireless communication framework more focuses on the meaning and utility of information at the application layer,
where semantic communication is a common example \cite{thomas2023neuro,9838470,wang2023adaptive}.
However, due to the broadcast nature of wireless communication, it faces security challenges such as eavesdropping and malicious attacks. 
Notably, AI-native communication presents new security requirements while ushering in new security opportunities \cite{yang2023secure}.

In 1949, Shannon proposed the cryptographic security mechanism from the information-theoretic perspective \cite{shannon1949communication}.
Based on Shannon cryptographic security, the authors in \cite{tung2023deep} and \cite{10328183} investigated semantic encryption design for information protection.
However, encryption schemes often assume that attackers have finite computation, which is being challenged as quantum computing develops.
Wyner \cite{wyner1975wire} further introduced the wire-tap channel based on Shannon confidential communications,
and enabled secure transmission at the physical layer using channel characteristics such as noise, fading, and interference \cite{10459057}.
In \cite{wang2024star}, the channel characteristics were used to enhance the semantic security.
However, eavesdroppers occupying favorable channels still pose challenges to semantic security at the physical layer.
In fact, goal-oriented semantic communications focus on task performance at the application layer, 
thus the application layer security needs further consideration.
The authors in \cite{10123081} and \cite{cheng2024knowledge} respectively utilized physical channel difference and knowledge difference to enhance the task security at the application layer.
However, there is a blank in the unified security metrics to comprehensively estimate semantic secure communications.

Resource allocation schemes can significantly enhance security efficiency, especially with limited resources \cite{9832831}, 
and semantic resource allocation further raises the allocation requirements from the task view.
Moreover, with the development of AI-native wireless networks, 
network intelligence and time-efficiency resource allocation also play an important role \cite{saad2024artificial}.
Recently, reinforcement learning (RL) based semantic resource schemes have received widespread attention 
due to their ability to fulfill the above-mentioned demands \cite{zheng2024energy}.
In \cite{wang2023adaptive}, the authors proposed an RL-enhanced adaptive resource allocation scheme, 
where the agent was able to be aware of different semantic representation requirements.
In \cite{9832831}, the RL agent selectively transmitted semantic blocks based on data importance.
Notably, intelligent semantic resource allocation that is aware of cross-layer security remains unexplored.

In this letter, we first define a unified semantic security metric, which aims to comprehensively estimate cross-layer security requirements.
Then, we propose the hierarchical AI-native semantic secure communication network with an RL-based semantic resource allocation scheme,
which can overcome the formulated cross-layer semantic security rate (CLSSR) maximization problem with the mixed integer nonlinear programming (MINLP).
The convergence of our proposed intelligent resource allocation is proved, and the effectiveness of our cross-layer semantic security (CLSS) design is validated through numerical experiments.

The remainder of this paper is organized as follows. 
Section II presents the semantic communication network. 
In Section III, the physical layer security and the application layer security are jointly considered, 
and the CLSSR maximization problem is formulated.
Section IV presents an intelligent resource allocation scheme.
In Section V, the simulation results are presented.
Finally, Section VI concludes this letter.

\section{Secure Semantic Communication}
\begin{figure}
  \centering
  \includegraphics[width=8.5cm]{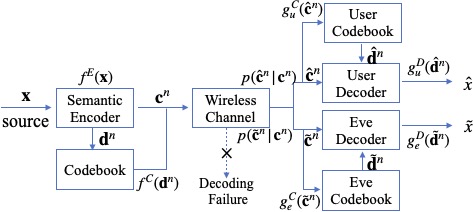}
  \caption{The Secure semantic communication Model.}
\end{figure}
The aim of this letter is to explore the security properties of the goal-oriented semantic communications, 
thus we discuss a generalized semantic communication system, independent of specific tasks and modalities, as shown in Fig. 1.
We first consider a memoryless source $\mathcal{X}$, where the independent and identically distributed (i.i.d.) sequence $\mathbf{x}\in\mathcal{X}$,
and the goal attention is the intrinsic semantics $\mathbf{d}$ in $\mathbf{x}$.
To obtain the semantics, there is a semantic encoder to extract the key information,
represented by $\mathbf{d}^n = f^{E}(\mathbf{x})$, where $\mathbf{d}^n$ is a sequence of $n$-length semantics.
To confuse the Eve with semantics, a user-perceivable learnable semantic noise $\Delta \mathbf{z}^n$ is added, 
where the noise-disturbed semantic is presented by ${\mathbf{d}^{\prime}}^n = \mathbf{d}^n + \Delta \mathbf{z}^n$.
The goal-related semantics are then encoded by the codebook \cite{10622764}, which is a common knowledge base and generates codeword indexes.
Here, we consider a codebook with binary indexes.
Let $\mathcal{C}^n=\left\{1,2,\cdots,M\right\}$ be a set of the  $n$-length codeword index, and $\mathbf{c}^n$ be a sequence of $n$-length index, $\mathbf{c}^n \in \mathcal{C}^n$.
Thus, the semantic code is obtained by $\mathbf{c}^n = f^C({\mathbf{d}^{\prime}}^n)$, and then transmitted over the wireless channel.

At the user, the received semantic code $\hat{\mathbf{c}}^n$ is obtained by the transition probability distribution $p(\hat{\mathbf{c}}^n|\mathbf{c}^n)$ over the user channel,
and we can recover the semantics by using the user codebook, written as $\hat{\mathbf{d}}^n = g_u^C(\hat{\mathbf{c}}^n)$.
The target goal is achieved with the semantic decoder at the user side $\hat{\mathbf{x}} = g_u^D(\hat{\mathbf{d}}^n)$.
At the Eve, the received semantic code $\tilde{\mathbf{c}}^n$ is obtained by the transition probability distribution $p(\tilde{\mathbf{c}}^n|\mathbf{c}^n)$ over the eavesdropping channel,
where both semantics transmission and semantics eavesdropping happens in the licensed spectrum.
The eavesdropped semantics is recovered by $\tilde{\mathbf{d}}^n = g_e^D(\tilde{\mathbf{c}}^n)$, and then the Eve speculates on the target goal $\tilde{\mathbf{x}} = g_e^D(\tilde{\mathbf{d}}^n)$.
The mappings describing the transmitter, receiver, and Eve are respectively written by
\begin{subequations}\label{eq1}
  \begin{align}
  f^C(f^{E}(\mathbf{x})+\Delta \mathbf{z}^n): \mathcal{X} \rightarrow \mathcal{C}^n, \\
  g_u^D(p(\hat{\mathbf{c}}^n|\mathbf{c}^n)):  \mathcal{C}^n \rightarrow  \widehat{\mathcal{X}}, \\
  g_e^D(p(\tilde{\mathbf{c}}^n|\mathbf{c}^n)): {\mathcal{C}}^n  \rightarrow \widetilde{\mathcal{X}}.
  \end{align}
\end{subequations}

The semantic secure communication scenario discussed in this letter is clarified. 
The encryption is not used at the physical and application layers, i.e. the open signal is communicated from the transmitter to the receiver.
The transmitter with multiple antennas transmits beamforming, 
which concentrates the signal in the direction of the legitimate user and reduces the radiation of the signal in other directions, 
thus suppressing the symbol quality of the Eves. 
At the application layer, the learnable semantic noise $\mathbf{z}^n$ is used to confuse the semantics since it is unnoticeable for the Eve. 
Moreover, the Eve can apply the statistical pattern to align the received semantics with the target semantics.
Due to the robustness of the semantics, it is still possible for the Eve to eavesdrop on some useful task information with the decoder $g_e^D$.

\section{Physical Layer Security and Application Layer Security}
In this section, 
we present a detailed analysis of security optimization targets at the physical layer and the application layer for semantic communications.
Then, a cross-layer semantic security metric is given.

\subsection{Physical Layer Security}
At the physical layer, 
the expectation of the user is to accurately receive the codeword $\hat{\mathbf{c}}^n$,
whereas the Eve attempts to eavesdrop on this codeword.
Thus, the security target at the physical layer can be achieved 
when the user successfully receives the accurate codeword while the Eve fails to recover the codeword with the incomplete eavesdropped information.

The security rate is applied to evaluate the security efficiency at the physical layer.
Specifically, the transmit beamforming from the transmitter is represented by $\mathbf{f}_u$.
Let $\mathbf{h}_u$ and $\mathbf{h}_e$ be the channel gain from the transmitter to the user and the channel gain from the transmitter to the Eve, respectively.  
The transmission rate of the user $u$ and the Eve $e$ can be respectively represented by
\begin{subequations}
  \begin{align}
  R_u=B \log _2\left(1+\frac{|\mathbf{h}_u\mathbf{f}_u|^2}{\sigma_u^2}\right),\\
  R_e=B \log _2\left(1+\frac{|\mathbf{h}_e\mathbf{f}_u|^2}{\sigma_e^2}\right),
\end{align}
\end{subequations}
where $\sigma_u^2$ and $\sigma_e^2$ are respectively the noise power at the user and the Eve.
Then, the security rate at the physical layer can be represented by
\begin{equation}
  R^{S}_{u,e}=\left(R_u-R_e\right)^{+},
\end{equation}
where $\left(x\right)^{+}$ means when $x$ is positive, it makes sense.

\subsection{Application Layer Security}
At the application layer, the goal is to complete the semantic task at the legal user with $\hat{\mathbf{d}}^n$, 
whereas preventing completing the task with eavesdropped semantics $\tilde{\mathbf{d}}^n$. 
Different from the symbol protection at the physical layer, the application layer security is strongly related to the task. 
Hence, we define a security metric called \textit{task security} at the application layer.

Firstly, the distortion $d(\mathbf{x}_1,\mathbf{x}_2)$ can access the goal completion degree, 
where $d$ evaluates the difference between $\mathbf{x}_1$ and $\mathbf{x}_2$ at the semantic level.
We apply the goal completion degree to represent the \textit{task reliability}.
Specifically, let the target goal of symbols $\mathbf{x}$ be $\mathbf{x}^{\mathrm{goal}}$.
The normalized semantic distortions related to the goal at the user $u$ and the Eve $e$ are respectively written by 
\begin{subequations}
  \begin{align}
  &G^n_u \triangleq \frac{d\left(\hat{\mathbf{x}}, \mathbf{x}^{\mathrm{goal}}\right)}{d^{\mathrm{max}}}, \\
  &G^n_e \triangleq \frac{d\left(\tilde{\mathbf{x}}, \mathbf{x}^{\mathrm{goal}}\right)}{d^{\mathrm{max}}},
\end{align}
where $d^{\mathrm{max}}$ represents the upper limit of the distortion metric $d(\cdot)$.
\end{subequations}
Then, jointly considering the physical layer security cases, the \textit{task security} can be written as 
\begin{equation}\label{ts}
  G_{u,e}^{S,n}=\left\{
\begin{aligned}
  &1-G^n_u & & \mathcal{H}_s,\\
  &0& & \text{otherwise},
\end{aligned}
\right.
\end{equation}
where $\epsilon$ is the threshold that can ensure task completion,
and $\mathcal{H}_s$ is the task security case.
\begin{definition}
  We define the \textit{task security} case in a logical-symbolic way, which can be represented by $\mathcal{H}_s: (G^n_u < \epsilon) \wedge (R^{S}_{u,e} > 0 \vee R^{S}_{u,e} = 0, G^n_e \ge \epsilon)$.
\end{definition}

Specifically, it is necessary to ensure that semantic tasks can be completed by the user, i.e., $G^n_u \le \epsilon$.
Then, we consider two cases that can ensure task security. 
The first case is the task security at the physical layer, i.e., $R^{S}_{u,e} > 0$, while the second case is the task security at the application layer even with semantic symbol leakage at the physical layer, i.e., $R^{S}_{u,e} = 0, G^n_e > \epsilon$.

\begin{remark}
  The fact that tasks can be accomplished is the basic requirement of security. 
  Then, with physical layer security, the application layer is certainly secure.
  While the physical layer security cannot be guaranteed, 
  the application layer can also guarantee the system security under task constraints.
  Hence, it is clear that \textit{task security} (\ref{ts}) relaxes the security constraint at the physical layer 
  for goal-oriented semantic communications.
  As for the threshold $\epsilon$, it differs in different tasks.
  For example, on a discrete-type task, such as image categorization, $\epsilon = 1$, 
  which means the user can always successfully classify, i.e., $G^n_u=0$, while the Eve fails, i.e., $G^n_e=1$.
  It is more complicated to set $\epsilon$ on a continuous-type task, such as image reconstruction and speech recognition,
  that expert experience can be required to determine $\epsilon$. For example, $\epsilon$ can be set to $0.01$ based on the metric of mean square error loss for the image reconstruction task. Similar to the discrete task, the semantic security case of continuous-type tasks can also be defined by using Definition 1.
\end{remark}

\subsection{Cross-Layer Semantic Security}
To jointly consider the semantic security at the physical layer and the application layer,
we propose a cross-layer semantic security rate (CLSSR) for the first time.
Specifically, we map the semantic security requirements of the application layer into the physical layer.
The cross-layer semantic-bit security efficiency can be defined as
\begin{equation}
  \Phi_{u,e}^{S,n} = \frac{G_{u,e}^{S,n}}{n \times {b}},
\end{equation}
where $b$ represents the number of bits for each codeword index, i.e., the bits required for each semantic symbol.
In fact, $n \times {b}$ is the total bit overhead at the physical layer for the semantic security achieved at the application layer.
Let the semantic symbol efficiency $\frac{1(G^{S,n})}{n}$ be the semantic unit ($sut$),
and then the unit of the $\Phi^{S,n}$ can be written as $sut/bit$.
It is evident that the $\Phi^{S,n}$ is the semantic gain of each bit and serves as the bridge between the physical layer and the application layer.
Hence, the CLSSR is defined by
\begin{equation}\label{clssr}
  \Omega_{u,e}^{S,n} =\left\{
\begin{aligned}
  &\Phi_{u,e}^{S,n} \times {R}_{u,e}^{S} & & R_{u,e}^{\mathrm{S}} > 0, \\
  &\Phi_{u,e}^{S,n} \times {R}_{u} & & R_{u,e}^{\mathrm{S}} = 0,
\end{aligned}
\right.
\end{equation}
where the unit of the $\Omega_{u,e}^{S,n}$ is $sut/s$.
In this way, the CLSSR $\Omega_{u,e}^{S,n}$ can effectively evaluate the security performance through both the physical layer and the application layer,
guiding the cross-layer semantic security design.

\begin{remark}
  Based on the task security case (Definition 1), the task security requirement $G^n_e > \epsilon$ at the application layer relaxes the strict symbol security at the physical layer.
  Hence, it can be observed from (\ref{clssr}) that the transmitter can abandon the physical layer security 
  in order to obtain the higher $\Omega_{u,e}^{S,n}$ since the symbol transmission rate ${R}_{u} \ge R_{u,e}^{\mathrm{S}}$.
  Notably, such rate increase plays an important role in real-time sensitive semantic tasks, thus ensuring \textbf{task timeliness}, which is one of the most important features of future intelligent wireless communication networks.
\end{remark}

Further considering a cooperative eavesdropping scenario with multiple eavesdroppers, 
in which case that the secure and insecure physical layer may occur at the same time,
the CLSSR at the user $u$ is written as 
\begin{equation}
  \Omega_{u}^{S,n} = \min_{e}\Omega_{u,e}^{S,n}.
\end{equation}
Here, (8) clearly demonstrates the necessity of considering the application layer security for semantic communications,
which assesses semantic security in a more comprehensive way. 

\begin{remark}
  In fact, although the length $n$ of the semantics can be any value theoretically, 
  in terms of engineering applications and many practical deep learning-based semantic communication networks, 
  $n$ tends to be a fixed value that is a multiple of two.
  Hence, we can abbreviate the codeword $\mathbf{c}^{n}$, semantic noise $\Delta \mathbf{z}^n$ and the CLSSR $\Omega_u^{S,n}$ respectively as $\mathbf{c}$, $\Delta \mathbf{z}$ and $\Omega_u^{S}$ since the number of semantic symbols in this paper is predefined.
\end{remark}

\subsection{CLSSR Maximization Problem Formulation}
We formulate the CLSSR maximization problem by jointly 
optimizing the transmit beamforming $\mathbf{f}_u$, bits $b$ for each semantic symbol representation in the codebook, and
the semantic decoder $g_u^D$ at the user along with user-perceivable learnable semantic noise $\Delta \mathbf{z}$,
which is given as
\begin{subequations}\label{opt}
  \begin{align}
   \mathbf{P}: &\max_{\mathbf{f}_u, b, g_u^D, \Delta \mathbf{z}} \Omega^{S}_u \\ 
   \text{ s.t. }  
   & \left\|\mathbf{f}_u\right\|^2 \leq \xi,\\
   & b_{\min} \leq b \leq b_{\max},\\
   & G^n_e > G^n_u.
  \end{align}
\end{subequations}
The constraint (\ref{opt}b) represents the power limit, where $\xi$ is the maximum transmit power of the transmitter.
The constraint (\ref{opt}c) presents the range of index bits, where $b_{\min}$ is the low bound, and $b_{\max}$ is the upper bound.
The constraint (\ref{opt}d) gives the minimum security requirement at the application layer, which guides the training of the user decoder $g_u^{C}$ and user-perceivable learnable semantic noise $\Delta \mathbf{z}$.

The allocation of the index bits $b$ can be considered a 0-1 matching problem, 
and the optimization of the user decoder $g_u^D$ and and user-perceivable learnable semantic noise $\Delta \mathbf{z}$ are coupled with the beamforming design $\mathbf{f}_u$ at the physical layer.
Hence, it is evident that (\ref{opt}) is a MINLP, which is proved to be a NP-hard problem.

\section{Proposed Policy Iteration RL-enhanced Intelligent Resource Allocation Scheme}
Due to the problem complexity, the dynamic radio environment, and the intelligence requirement, 
a policy iteration RL-based intelligent resource allocation scheme is proposed.
The Markov decision processing (MDP) is utilized to model the semantic communication network.

\subsection{MDP Model}
Let $\mathcal{A}$ be the action space.
The action includes $\mathbf{f}_u$ and $b$, represented by $a^t = \{\mathbf{f}^t_u,b^t\}$ at the timestamp $t$, $a^t \in \mathcal{A}$.
The state includes the current codewords $\mathbf{c}$, legal channel gain $\mathbf{h}_{u}$ and the illegal channel gain $\mathbf{h}_e$,
represented by $s^t = \{{\mathbf{c}}^t,\mathbf{h}^t_u,\mathbf{h}_e^t\}$ at the timestamp $t$, $s^t \in \mathcal{S}$.
We consider steady-state communications with the stationary source between the transmitter and the user.
Hence, the source transition matrix $\mathrm{P}^S$ is fixed, where the element $p^S(\mathbf{x}_2|\mathbf{x}_1)$ represents the probability of the source $\mathbf{x}_2$ following the source $\mathbf{x}_1$.
Given $f^{C}$, we have $p^S(\mathbf{c}_2|\mathbf{c}_1)=p^S(\mathbf{x}_2|\mathbf{x}_1)$.
As $\mathbf{h}_u$ and $\mathbf{h}_e$ are the outcomes of the action $a$,
the state transition probability follows the source transition matrix. 
Let $\mathrm{P}^T$ be the state transition matrix, where the element $p^T(s^{\prime}|s^t,a^t)=p^S(\mathbf{c}^{t+1}|\mathbf{c}^{t})$ and $\mathrm{P}^T \sim \mathrm{P}^S$.

\begin{remark}
The bit allocation is considered as a discrete matching problem, 
while the transmit beamforming optimization can be also regarded as the limited action choices under the retained accuracy.
Hence, the size of the action space $|\mathcal{A}|$ is finite. Let $\mathcal{S}$ be the state space.
  The state $s$ is coupled with the $a^{t}$ and ${\mathbf{c}}^t$, thus the state space size $|\mathcal{S}|$ is also finite.
\end{remark}

The CLSSR serves as the reward function, represented by $r^t = {\Omega_{u}^{S,t}}$, where ${\Omega_{u}^{S,t}}$ is the achievable CLSSR at the timestamp $t$.
However, the semantic coding is also responsible for the CLSSR, 
and this reward can be severely disrupted by the instability of the machine learning-driven $f^C$, $g_u^D$, $g_e^D$, 
resulting in the failure of the intelligent agent to find an optimal policy.
Here, we make two improvements. The first one is to redesign the reward function $r^t$, which is rewritten as
$r^t = {\Omega_{u}^{S,t}}+\sum_{e}(|\mathbf{h}_u\mathbf{f}_u|^2-|\mathbf{h}_e\mathbf{f}_e|^2)$, which aims to guide the agent to enhance the channel gain at the user while suppressing the channel gain at the Eves.
The second one is to pretrain the semantic coding networks. 
Harmonized training the encoder $f^C$ and the user decoder $g_u^D$, 
we have the pre-trained $f^{*C}$ and $g_u^{*D}$ along with the Eve decoder $g_e^{*D} = g_u^{*D}$, which is the worse eavesdropping case. Based on $g_u^{*D}$, we obtain the noise-enhanced user security decoder $\hat{g}_u^{*D}$ by adding the user-perceivable semantic noise $\Delta \mathbf{z}$,
based on which the whole semantic coding networks are tined with the policy of the intelligent agent.

\subsection{The Convergence of Proposed RL for CLSSR Maximization}
Let the policy be $\pi$,
and the policy value is obtained by
\begin{equation}
  \begin{split}
  v_\pi(s)&=\sum_a \pi(a \mid s) \sum_{s^{\prime}} p(s^{\prime} \mid s, a)(r + \gamma v_\pi\left(s^{\prime}\right)) \\
  &= \sum_{{\mathbf{c}^{\prime}}} p^S({\mathbf{c}^{\prime}}|{\mathbf{c}})(r + \gamma v_\pi\left(s^{\prime}\right)),
\end{split}
\end{equation}
where the $\gamma \in (0,1)$ is the reward discount factor.
The state-action value can be obtained by
\begin{equation}
  q_\pi(s, a)=\sum_{s^{\prime}} p\left(s^{\prime} \mid s, a\right)\left(r+\gamma v_\pi\left(s^{\prime}\right)\right).
\end{equation}

\begin{lemma}
  The policy $\pi$ is iteratively improved.
\end{lemma}
\begin{proof}
  We have the finite $|\mathcal{A}|$ and $|\mathcal{S}|$ as stated in Remark. 2.
Given a state $s$, the updated policy satisfies
$\pi^{\prime}(s)=\arg \max _a q_\pi(s, a)$. Hence, we have
\begin{equation}
  \begin{split}
  v_\pi(s) & \leq q_\pi\left(s, \pi^{\prime}(s)\right)\\
  & = \sum_{s^{\prime}} p\left(s^{\prime} \mid s, \pi^{\prime}(s)\right)\left(r+\gamma v_{\pi}\left(s^{\prime}\right)\right) \\
  & = \mathbb{E} \left[ r+\gamma v_{\pi}\left(s^{\prime}\right) \mid s, \pi^{\prime}(s)\right] \\
  & \leq \mathbb{E} \left[ r+\gamma q_\pi\left(s^{\prime}, \pi^{\prime}(s^{\prime})\right) \mid s, \pi^{\prime}(s)\right]\\
  & = \mathbb{E} \left[ r+\gamma \mathbb{E}\left[r^{\prime}+\gamma v_\pi(s^{\prime\prime})\mid s^{\prime}, \pi^{\prime}(s^{\prime})\right] \mid s, \pi^{\prime}(s)\right]\\
  & = \mathbb{E} \left[ r+\gamma r^{\prime} + \gamma^2 v_\pi(s^{\prime\prime}) \mid s, \pi^{\prime}(s)\right]\\
  & \leq \mathbb{E} \left[ r+\gamma r^{\prime} + \gamma^2 r^{\prime\prime} + \gamma^3 v_\pi(s^{\prime\prime\prime}) \mid s, \pi^{\prime}(s)\right]\\
  & \cdots \\
  & = v_{\pi^{\prime}}(s)
\end{split}
\end{equation}
The alternating iterations of $v_\pi$ and $q_\pi$ are omitted, and $v_\pi(s) \leq v_{\pi^{\prime}}(s)$ proves the Lemma 1.
\end{proof}

Given the state $s$ and policy $q_\pi$, the action is obtained by
\begin{equation}
  \mathrm{a}=\underset{a}{\arg\max} q_\pi(s, a).
\end{equation} 

\begin{theorem}
  The optimal policy $\pi^*$ can be iteratively found.
\end{theorem}

\begin{proof}
Considering the $|\mathcal{A}|$ and $|\mathcal{S}|$ are finite,
our considered MDP model follows the Bellman's optimality equation \cite{sutton2018reinforcement} with the Contracting mapping theorem.
Thus, there exists one and only one optimal policy $\pi^*$.
As the policy space is finite, let the policy set of multiple time steps be $\phi = \{\pi_1,\pi_2,\cdots,\pi_t,\cdots,\pi_T\}$.
Since the improvement nature of the policy $\pi$ is proved in Lemma 1, 
there exists $m$, satisfying $\pi_m = \pi^*$ and $\pi_t = \pi^*, \forall t \ge m$.
The convergence is proved.
\end{proof}

\section{Simulation And Analyses}
In this section, 
several numerical experiments are implemented to validate the necessity and effectiveness of the CLSS design.
If not stated additionally, the simulation parameters are as follows.
The base station and user locations are fixed, respectively at (0, 0) and (500, 500), and two Eves are respectively located at (400, 450) and (800, 300).
The transmit power is set as $TP = 20 \mathrm{~dBm}$, the bandwidth is $B = 1 \mathrm{~MHz}$, and the Rayleigh fading is considered.
The noise power density is set as $N_0 = -174 \mathrm{~dBm/Hz}$, 
and the noise factor $\tau$ is used to explore different nature noise power $P_0$, i.e., 
the ground noise power $P$ satisfies $P = \tau P_0$
Without loss of generality, we adopt the semantic coding framework in \cite{10123081} and the adaptive codewords in \cite{wang2023adaptive}.
Specifically, the image modality and reconstruction task are considered, where the perceptual loss is used as $d(\cdot)$.
We consider two typical antenna configurations, multiple input multiple output (MIMO) and multiple input single output (MISO),
where $2 \times 2$ and $2 \times 1$ Alamouti space-time diversity techniques are applied.
The modulation method is set to $4\mathrm{-PSK}$.
For better comparison, the traditional PL-SS without AL-SS design (\ref{eq1}a) is introduced.

To put further emphasis on the advantage of our designed CL-SS method, 
we define a new metric called \textit{task reliability} that quantifies \textit{task security} and \textit{task timeliness} that are achieved at the same time.
Specifically, \textit{task reliability} is mathematically defined as the percentage of the number of tasks $\text{T}_{\text{security } \cap \text{ timeliness}}$ that achieve both security and timeliness to the total number of performed tasks ${\text{T}_\text{total}}$, which is represented by
\begin{equation}
    \text{Task Reliability} = \frac{\text{T}_{\text{security } \cap \text{ timeliness}}}{\text{T}_\text{total}} \times 100 \%.
\end{equation}
The task reliability convergence of our proposed CL-SS and the conventional PL-SS versus the training episodes with $\tau=2$ is presented in Fig. 2.
It demonstrates that our proposed CL-SS method can satisfy the reliability requirements. i.e., achieving 100\% task reliability, with a convergence speed 300\% faster than that of the conventional PL-SS method.
This is due to the fact that in the CL-SS method, the application layer can ensure the system security with the model differences and semantic noise design at a higher symbol rate $R_u$, 
thus relaxing the physical layer security, as analyzed in (\ref{clssr}).
In contrast, the traditional PL-SS method relies only on the physical layer optimization through beamforming, 
which simultaneously requires inhibition of eavesdroppers from obtaining useful information and ensures the task timeliness.
It is evident that the overhead of employing a pre-trained semantic secure coding model for the application layer
is significantly less than that of resource optimization for the physical layer, such as multiple-antennas design.
Therefore, it is a promising solution to develop novel security techniques at the application layer and cross-layer security techniques for semantic communications.

\begin{figure}
  \centering
  \includegraphics[width=8.5cm]{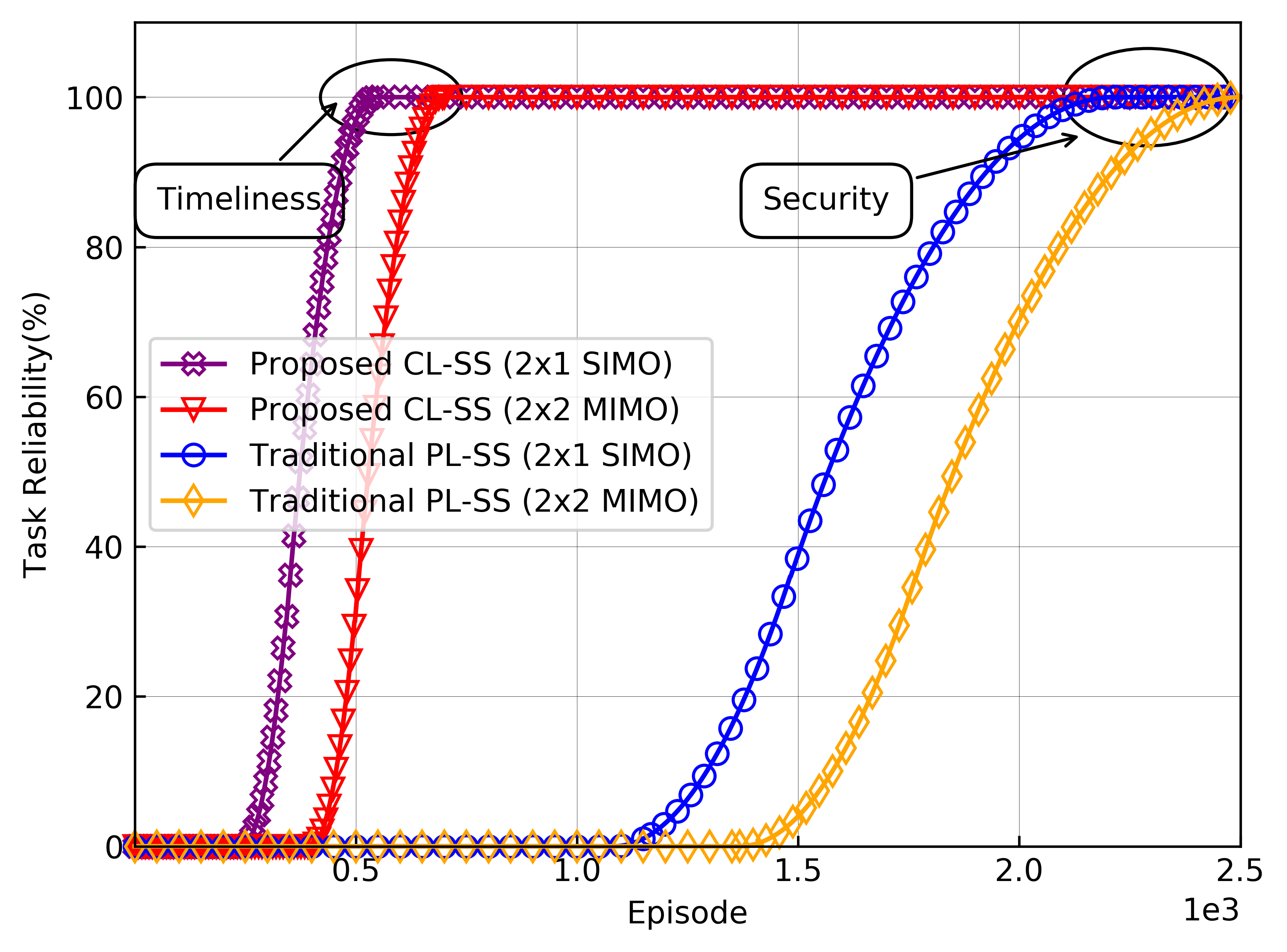}
  \caption{The reliability convergence of semantic security models.}
\end{figure}

\begin{figure}
  \centering
  \includegraphics[width=8.5cm]{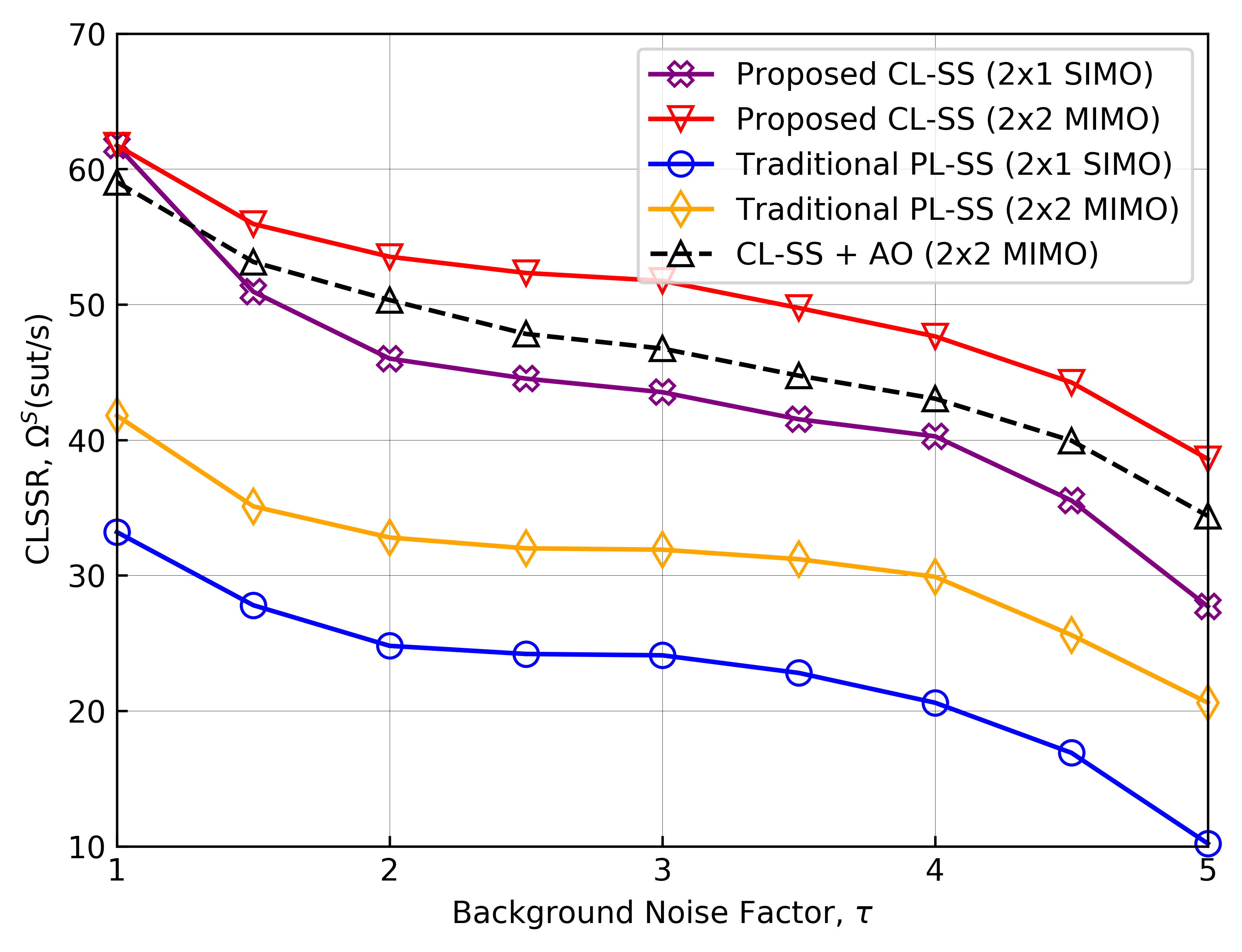}
  \caption{The achievable CLSSR of semantic secure communication models.}
\end{figure}

Fig. 3 demonstrates the achievable CLSSR of our proposed CL-SS and the conventional PL-SS versus background noise.
We can observe that our proposed CL-SS method with MIMO can achieve up to 85\% improvement when $\tau=1$ and 295\% improvement when $\tau=5$
compared to the traditional PL-SS method with MIMO in terms of CLSSR.
This observation shows the superior performance of our proposed CCSSR since the CL-SS method
can find optimal solutions in prioritizing the physical layer security or relaxing the physical layer security into the application layer security.
Since the beamforming design can enhance the signal quality in a particular direction with the task eavesdropping suppressed at the Eves, 
a flat performance degradation can be seen when $\tau$ is from 2 to 3. 
This implies that there is a task gain of the beamforming design for semantic communication when the background noise is strong.
Moreover, the alternating optimization (AO)-based conventional scheme is introduced for better comparison. Due to the synergistic effect of the intelligence-native RL scheme with the AI-driven semantic coding, 
our proposed RL-based optimization scheme performs better not only in CLSSR but also in AI ecosystem support for practical applications.

\section{Conclusion}
This paper proposed to consider the security of semantic communications from a joint physical layer and application layer perspective for the first time.
We first defined the unified semantic security metric, CLSSR, to comprehensively estimate the cross-layer security requirements. Then we formulated the CLSSR maximization problem by jointly optimizing the transmit beamforming, bits for each semantic symbol representation, the semantic decoder at the user, and the user-perceivable learnable semantic noise.
To address the CLSSR maximization problem with MINLP,
we proposed a policy iteration RL-based semantic resource allocation scheme.
The convergence of our proposed intelligent resource allocation was proved, 
and our proposed CLSS method outperformed the traditional PL-SS in terms of the task reliability and CLSSR.

\bibliography{IEEEabrv,reference}

\begin{thebibliography}{10}
\providecommand{\url}[1]{#1}
\csname url@samestyle\endcsname
\providecommand{\newblock}{\relax}
\providecommand{\bibinfo}[2]{#2}
\providecommand{\BIBentrySTDinterwordspacing}{\spaceskip=0pt\relax}
\providecommand{\BIBentryALTinterwordstretchfactor}{4}
\providecommand{\BIBentryALTinterwordspacing}{\spaceskip=\fontdimen2\font plus
\BIBentryALTinterwordstretchfactor\fontdimen3\font minus
  \fontdimen4\font\relax}
\providecommand{\BIBforeignlanguage}[2]{{%
\expandafter\ifx\csname l@#1\endcsname\relax
\typeout{** WARNING: IEEEtran.bst: No hyphenation pattern has been}%
\typeout{** loaded for the language `#1'. Using the pattern for}%
\typeout{** the default language instead.}%
\else
\language=\csname l@#1\endcsname
\fi
#2}}
\providecommand{\BIBdecl}{\relax}
\BIBdecl

\bibitem{saad2024artificial}
W.~Saad, O.~Hashash, C.~K. Thomas, C.~Chaccour, M.~Debbah, N.~Mandayam, and
  Z.~Han, ``Artificial general intelligence {(AGI)}-native wireless systems: A
  journey beyond {6G},'' \emph{arXiv preprint arXiv: 2405.02336}, 2024.

\bibitem{thomas2023neuro}
C.~K. Thomas and W.~Saad, ``Neuro-symbolic causal reasoning meets signaling
  game for emergent semantic communications,'' \emph{IEEE Trans. Wireless
  Commun.}, vol.~23, no.~5, pp. 4546--4563, 2024.

\bibitem{9838470}
F.~Zhou, Y.~Li, M.~Xu, L.~Yuan, Q.~Wu, R.~Q. Hu, and N.~Al-Dhahir, ``Cognitive
  semantic communication systems driven by knowledge graph: Principle,
  implementation, and performance evaluation,'' \emph{IEEE Trans. Commun.},
  vol.~72, no.~1, pp. 193--208, 2024.

\bibitem{wang2023adaptive}
L.~Wang, W.~Wu, F.~Zhou, Z.~Yang, Z.~Qin, and Q.~Wu, ``Adaptive resource
  allocation for semantic communication networks,'' \emph{IEEE Trans. Commun.},
  vol.~72, no.~11, pp. 6900--6916, 2024.

\bibitem{yang2023secure}
Z.~Yang, M.~Chen, G.~Li, Y.~Yang, and Z.~Zhang, ``Secure semantic
  communications: Fundamentals and challenges,'' \emph{IEEE Network}, 2024, to
  be published.

\bibitem{shannon1949communication}
C.~E. Shannon, ``Communication theory of secrecy systems,'' \emph{Bell System
  Techn. J.}, vol.~28, no.~4, pp. 656--715, 1949.

\bibitem{tung2023deep}
T.-Y. Tung and D.~G{\"u}nd{\"u}z, ``Deep joint source-channel and encryption
  coding: Secure semantic communications,'' in \emph{Proc. IEEE Int. Conf.
  Commun. (ICC)}, 2023, pp. 5620--5625.

\bibitem{10328183}
X.~Liu, G.~Nan, Q.~Cui, Z.~Li, P.~Liu, Z.~Xing, H.~Mu, X.~Tao, and T.~Q.
  S.~Quek, ``{SemProtector}: A unified framework for semantic protection in
  deep learning-based semantic communication systems,'' \emph{IEEE Commun.
  Mag.}, vol.~61, no.~11, pp. 56--62, 2023.

\bibitem{wyner1975wire}
A.~D. Wyner, ``The wire-tap channel,'' \emph{Bell System Techn. J.}, vol.~54,
  no.~8, pp. 1355--1387, 1975.

\bibitem{10459057}
L.~Liang, X.~Li, H.~Huang, Z.~Yin, N.~Zhang, and D.~Zhang, ``Securing
  multidestination transmissions with relay and friendly interference
  collaboration,'' \emph{IEEE Internet Things J.}, vol.~11, no.~10, pp.
  18\,782--18\,795, 2024.

\bibitem{wang2024star}
Y.~Wang, W.~Yang, P.~Guan, Y.~Zhao, and Z.~Xiong, ``{STAR-RIS}-assisted privacy
  protection in semantic communication system,'' \emph{IEEE Trans. Veh.
  Technol.}, 2024, to be published.

\bibitem{10123081}
M.~Zhang, Y.~Li, Z.~Zhang, G.~Zhu, and C.~Zhong, ``Wireless image transmission
  with semantic and security awareness,'' \emph{IEEE Wireless Commun. Lett.},
  vol.~12, no.~8, pp. 1389--1393, 2023.

\bibitem{cheng2024knowledge}
S.~Cheng, X.~Zhang, Y.~Sun, Q.~Cui, and X.~Tao, ``Knowledge discrepancy
  oriented privacy preserving for semantic communication,'' \emph{IEEE Trans.
  Veh. Technol.}, 2024, to be published.

\bibitem{9832831}
Y.~Wang, M.~Chen, T.~Luo, W.~Saad, D.~Niyato, H.~V. Poor, and S.~Cui,
  ``Performance optimization for semantic communications: An attention-based
  reinforcement learning approach,'' \emph{IEEE J. Sel. Areas Commun.},
  vol.~40, no.~9, pp. 2598--2613, Sep. 2022.

\bibitem{zheng2024energy}
J.~Zheng, B.~Du, H.~Du, J.~Kang, D.~Niyato, and H.~Zhang, ``Energy-efficient
  resource allocation in generative {AI}-aided secure semantic mobile
  networks,'' \emph{IEEE Trans. Mob. Comput.}, 2024, to be published.

\bibitem{10622764}
L.~Wang, W.~Wu, F.~Zhou, F.~Tian, Q.~Wu, and W.~Saad, ``A unified hierarchical
  semantic knowledge base for multi-task semantic communication,'' in
  \emph{Proc. IEEE Int. Conf. Commun. (ICC)}, 2024, pp. 2937--2943.

\bibitem{sutton2018reinforcement}
R.~S. Sutton and A.~G. Barto, \emph{Reinforcement learning: An
  introduction}.\hskip 1em plus 0.5em minus 0.4em\relax MIT press, 2018.

\end{thebibliography}

\end{document}